\newif\ifarxiv
\arxivtrue 

\ifarxiv

\documentclass[11pt,letterpaper]{article}
\usepackage{typearea}
\usepackage{setspace}

\usepackage{fullpage}

\usepackage{comment} 
\usepackage{bbm}
\usepackage{framed} 
\usepackage{url} 
\usepackage{complexity}
\usepackage{booktabs}
\usepackage{amsmath,amssymb}
\usepackage{float}

\usepackage{amsthm}
\usepackage{thmtools} 
\usepackage{thm-restate}
\usepackage{nicefrac}
\usepackage{calc}
\usepackage{enumerate}
\usepackage{enumitem}
\usepackage[usenames,dvipsnames]{xcolor}

\usepackage[ruled,vlined,linesnumbered]{algorithm2e}
\usepackage{forloop}

\usepackage{graphicx}
\usepackage[font=footnotesize,labelfont=bf]{subcaption}
\usepackage[font=footnotesize,labelfont=bf]{caption}
\usepackage[nobreak=true]{mdframed}
\usepackage{appendix}
\usepackage[noend]{algpseudocode}
\usepackage[colorinlistoftodos]{todonotes}
\usepackage{xr}
\usepackage{array}
\usepackage{xspace}
\definecolor{ForestGreen}{rgb}{0.1333,0.5451,0.1333}
\definecolor{DarkRed}{rgb}{0.8,0,0}
\definecolor{Red}{rgb}{1,0,0}
\usepackage[linktocpage=true,
pagebackref=true,colorlinks,
linkcolor=DarkRed,citecolor=ForestGreen,
bookmarks,bookmarksopen,bookmarksnumbered]{hyperref}

\usepackage[capitalise]{cleveref}
\crefrangelabelformat{enumi}{#3#1#4--#5#2#6}
\usepackage{parskip}
\usepackage{chngcntr}
\usepackage{mathtools,stackengine}
\usepackage{multirow}

\stackMath
\newcommand{\stackGeq}[1]{%
	\setbox0=\hbox{${}\mathrel{\stackon[-1pt]{\geq}{\scriptstyle\text{#1\strut}}}{}$}
	\xdef\tmpwd{\dimexpr\the\wd0\relax}
	\kern.5\tmpwd\mathclap{\box0}&\kern.5\tmpwd
}

\usepackage{nameref}

\usepackage{tikz}
\usetikzlibrary{patterns}

\allowdisplaybreaks

\newcommand{\one}{{\sc Type 1}\xspace}
\newcommand{\two}{{\sc Type 2}\xspace}
\newcommand{\three}{{\sc Type 3}\xspace}

\newcommand{\eat}[1]{}

\newcommand{\ins}{{\tt insert}}

\newcommand\eps{\varepsilon}

\DeclarePairedDelimiterX{\expectarg}[1]{[}{]}{%
	\ifnum\currentgrouptype=16 \else\begingroup\fi
	\activatebar#1
	\ifnum\currentgrouptype=16 \else\endgroup\fi
}

\DeclarePairedDelimiterX{\nicesetarg}[1]{\{}{\}}{%
	\ifnum\currentgrouptype=16 \else\begingroup\fi
	\activatebar#1
	\ifnum\currentgrouptype=16 \else\endgroup\fi
}

\newcommand{\innermid}{\nonscript\;\delimsize\vert\nonscript\;}
\newcommand{\activatebar}{%
	\begingroup\lccode`\~=`\|
	\lowercase{\endgroup\let~}\innermid 
	\mathcode`|=\string"8000
}

\newcommand\opt{\mathsf{opt}\xspace}

\counterwithin{equation}{section}

\usepackage{eqparbox}

\theoremstyle{plain}

\newtheorem{theorem}{Theorem}[section]

\newtheorem{lemma}[theorem]{Lemma}

\newtheorem{claim}[theorem]{Claim}

\newtheorem{corollary}[theorem]{Corollary}

\newtheorem{remark}[theorem]{Remark}

\newlength{\continueindent}
\usepackage{etoolbox}
\makeatletter
\newcommand*{\ALG@customparshape}{\parshape 2 \leftmargin \linewidth \dimexpr\ALG@tlm+\continueindent\relax \dimexpr\linewidth+\leftmargin-\ALG@tlm-\continueindent\relax}
\apptocmd{\ALG@beginblock}{\ALG@customparshape}{}{\errmessage{failed to patch}}
\makeatother

\makeatletter
\def\thm@space@setup{%
	\thm@preskip=\parskip \thm@postskip=0pt
}
\makeatother

\usepackage{etoolbox}
\usepackage{tikz}
\usetikzlibrary{tikzmark}
\usetikzlibrary{calc}

\errorcontextlines\maxdimen

\newcommand{\ALGtikzmarkcolor}{black}
\newcommand{\ALGtikzmarkextraindent}{4pt}
\newcommand{\ALGtikzmarkverticaloffsetstart}{-.5ex}
\newcommand{\ALGtikzmarkverticaloffsetend}{-.5ex}
\makeatletter
\newcounter{ALG@tikzmark@tempcnta}

\newcommand\ALG@tikzmark@start{%
	\global\let\ALG@tikzmark@last\ALG@tikzmark@starttext%
	\expandafter\edef\csname ALG@tikzmark@\theALG@nested\endcsname{\theALG@tikzmark@tempcnta}%
	\tikzmark{ALG@tikzmark@start@\csname ALG@tikzmark@\theALG@nested\endcsname}%
	\addtocounter{ALG@tikzmark@tempcnta}{1}%
}

\def\ALG@tikzmark@starttext{start}
\newcommand\ALG@tikzmark@end{%
	\ifx\ALG@tikzmark@last\ALG@tikzmark@starttext
	\else
	\tikzmark{ALG@tikzmark@end@\csname ALG@tikzmark@\theALG@nested\endcsname}%
	\tikz[overlay,remember picture] \draw[\ALGtikzmarkcolor] let \p{S}=($(pic cs:ALG@tikzmark@start@\csname ALG@tikzmark@\theALG@nested\endcsname)+(\ALGtikzmarkextraindent,\ALGtikzmarkverticaloffsetstart)$), \p{E}=($(pic cs:ALG@tikzmark@end@\csname ALG@tikzmark@\theALG@nested\endcsname)+(\ALGtikzmarkextraindent,\ALGtikzmarkverticaloffsetend)$) in (\x{S},\y{S})--(\x{S},\y{E});%
	\fi
	\gdef\ALG@tikzmark@last{end}%
}

\apptocmd{\ALG@beginblock}{\ALG@tikzmark@start}{}{\errmessage{failed to patch}}
\pretocmd{\ALG@endblock}{\ALG@tikzmark@end}{}{\errmessage{failed to patch}}
\makeatother

\algblock[with]{With}{EndWith}
\algblockdefx[With]{With}{EndWith}%
[1]{\textbf{with} #1 \textbf{do}}%
{}

\makeatletter
\ifthenelse{\equal{\ALG@noend}{t}}%
{\algtext*{EndWith}}
{}%
\makeatother

\newcommand{\cE}{{\cal E}}

\author{
 {Yossi Azar\thanks{Department of Computer Science, Tel-Aviv University,
    Tel-Aviv, Israel. Email: {\tt azar@tauex.tau.ac.il}.}}
 \and
 {Debmalya Panigrahi\thanks{Department of Computer Science, Duke University, Durham, NC. Email: {\tt debmalya@cs.duke.edu}.}}
 \and
 {Or Vardi\thanks{Department of Computer Science, Tel-Aviv University,
    Tel-Aviv, Israel. Email: {\tt orvardi@mail.tau.ac.il}.}} 
}

\date{}

\title{Nearly Tight Bounds for the Online Sorting Problem}

\fi

\begin{document}

\maketitle

\begin{abstract}
   In the online sorting problem, a sequence of $n$ numbers in $[0, 1]$ (including $\{0,1\}$) have to be inserted in an array of size $m \ge n$ so as to minimize the sum of absolute differences between pairs of numbers occupying consecutive non-empty cells. Previously, Aamand {\em et al.} (SODA 2023) gave a deterministic 
$2^{\sqrt{\log n} \sqrt{\log \log n + \log (1/\varepsilon)}}$-competitive algorithm when $m = (1+\varepsilon) n$ for any $\varepsilon \ge \Omega(\log n/n)$. They also showed a lower bound: with $m = \gamma n$ space, the competitive ratio of any deterministic algorithm is at least $\nicefrac{1}{\gamma}\cdot\Omega(\log n / \log \log n)$. This left an exponential gap between the upper and lower bounds for the problem.

In this paper, we bridge this exponential gap and almost completely resolve the online sorting problem. First, we give a deterministic $O(\log^2 n / \varepsilon)$-competitive algorithm with $m = (1+\varepsilon) n$, for any $\varepsilon \ge \Omega(\log n / n)$. Next, for $m = \gamma n$ where $\gamma = [O(1), O(\log^2 n)]$, we give a deterministic $O(\log^2 n / \gamma)$-competitive algorithm. In particular, this implies an $O(1)$-competitive algorithm with $O(n \log^2 n)$ space, which is within an $O(\log n\cdot \log \log n)$ factor of the lower bound of $\Omega(n \log n / \log \log n)$. Combined, the two results imply a close to optimal tradeoff between space and competitive ratio for the entire range of interest: specifically, an upper bound of $O(\log^2 n)$ on the product of the competitive ratio and $\gamma$ while the lower bound on this product is $\Omega(\log n / \log\log n)$. We also show that these results can be extended to the case when the range of the numbers is not known in advance, for an additional $O(\log n)$ factor in the competitive ratio.
\end{abstract}

\pagenumbering{gobble}
\clearpage
\pagenumbering{arabic}
\setcounter{page}{1}

\section{Introduction}\label{sec:intro}Sorting is a classic problem in algorithm design, where given an array of $n$ elements, the algorithm is required to reorder them in a non-decreasing (or non-increasing) sequence. Recently, a natural online variant of this problem was introduced by Aamand {\em et al.}~\cite{AamandA0K23} where the elements are not known in advance but arrive over time and have to be inserted in a way that achieves as close to a sorted order as possible. Formally, we are given an online sequence of numbers $a_1, a_2, \ldots, a_n$, where $a_t \in [0, 1]$ for all $t \in [n]$, and the algorithm has to insert $a_t$ at time $t$ in a given array $A$ with $m$ cells. Clearly, $m$ must be at least $n$ to avoid collisions, but in general, it can be larger than $n$; we denote $\gamma := m/n$. The cost of the solution $c$ is given by the sum of absolute differences between consecutive non-empty cells of the array. Clearly, the cost of the optimal (sorted) sequence is the difference between the largest and smallest of the values. Following Aamand {\em et al.}~\cite{AamandA0K23}, we assume that a pair of sentinel values $0$ and $1$ are present and occupy the first and last cells of the array -- this fixes the cost of the optimal solution at $1$. (We also consider the case when the range of elements is not known in advance later in the paper.) Given this setup, the main question is:
\begin{quote}
{\em
``what is the best tradeoff between the size of the array $m$ and the cost of the online solution $c$, as a function of the number of elements $n$?''
}
\end{quote}

Tight upper and lower bounds on the cost are known when the array has no extra space, i.e., $m=n$. In this case, there is a deterministic algorithm that achieves a cost of $O(\sqrt{n})$ and this is the best achievable~\cite{AamandA0K23}. (In follow up work~\cite{AbrahamsenB0K024}, this bound was shown to be tight for randomized algorithms as well.) The situation is less clear when $m > n$, i.e., $\gamma > 1$. For $\gamma = 1+\eps$ (for any $\eps = \Omega(\log n / n)$), Aamand {\em et al.}~\cite{AamandA0K23} showed that there is a deterministic algorithm with cost $c = 2^{\sqrt{\log n}\sqrt{\log\log n + \log(1/\eps)}}$. In contrast, the best lower bound states that $\gamma\cdot c = \Omega(\log n / \log \log n)$~\cite{AamandA0K23}. In particular, if $\gamma = 1+\eps$, then the lower bound on the cost is $\approx \log n$ 
but the upper bound is $\approx 2^{\sqrt{\log n}}$. Bridging this exponentially large gap was an interesting question left open in \cite{AamandA0K23}. 

Another interesting question is whether we can use more space to reduce the cost of the solution. It is easy to show that if we insist on a cost of $1$, i.e., that the array be sorted, then $m$ must be exponential in $n$. But, if we allow $O(1)$ cost, then a simple argument already shows that $O(n^2)$ space is sufficient. For this, we assign $n$ contiguous cells in the array to each multiple of $1/n$. To insert $a_t$, we go to the next available cell allocated to the multiple of $1/n$ that is closest to $a_t$ (breaking ties arbitrarily). This has cost $\le 2$ because any pair of elements in adjacent non-empty cells differ by $\le 2/n$. But, is $O(n^2)$ space necessary for $O(1)$ cost, or can we use even less space? Note that the lower bound of $\gamma\cdot c = \Omega(\log n / \log \log n)$ only asserts that $\Omega(n \log n / \log \log n)$ space is necessary for $O(1)$ cost. More generally, can we get a smooth tradeoff of space used and cost incurred, between the two extremes of $(1+\eps) n$ space and $O(1)$ cost? 

\paragraph{Our Results.}
In this paper, we almost completely resolve the online sorting problem for the entire range between $(1+\eps) n$ space and $O(1)$ error. 
\begin{itemize}
    \item[-] For any $\eps \in [O(\log n/n), O(1)]$, we give a deterministic algorithm with a cost of $O(\log^2 n / \eps)$ for an array with $(1+\eps) n$ space. 
    
    Note that this result exponentially improves over the best known previous bound of $c = 2^{\sqrt{\log n}\sqrt{\log\log n + \log(1/\eps)}}$ due to Aamand {\em et al.}~\cite{AamandA0K23} and is within a quadratic factor of the lower bound of $\gamma\cdot c = \Omega(\log n / \log \log n)$.
    
    \item[-] For any $\gamma \in [O(1), O(\log^2 n)]$, we give a deterministic algorithm with a cost of $O(\log^2 n / \gamma)$ for an array with $\gamma n$ space. In particular, this implies an $O(n \log^2 n)$-space, $O(1)$-cost algorithm for this problem. 
    
    To the best of our knowledge, there is no previous result that achieves $O(1)$ cost, although as we outlined earlier, this is easy to achieve with $O(n^2)$ space. So, our result should be seen as improving the space bound for $O(1)$ cost from $O(n^2)$ to $O(n\log^2 n)$, which is within $\approx \log n$ factor of the lower bound of $\Omega(n \log n / \log \log n)$. In addition, we provide a smooth tradeoff between these two extremes of $(1+\eps) n$ space and $O(1)$ cost, and show an upper bound of $\gamma\cdot c = O(\log^2 n)$. Coupled with the lower bound of $\gamma\cdot c = \Omega(\log n / \log \log n)$, this shows that $\Omega(\log n / \log \log n) \le \gamma\cdot c \le O(\log^2 n)$ for the entire range of costs $c \in [O(1), O(\log^2 n)]$.
\end{itemize}

We also consider the more general setting when the range of the elements is not known in advance. In this case, the {\em competitive ratio} of an algorithm is given by the ratio between the algorithm's objective and the range of the elements (i.e. the optimal objective). Note that the competitive ratio corresponds to cost for unit range. We recover our results for unit range with the same space bounds in this more general setting, while losing an additional factor of $O(\log n)$ in the competitive ratio. We give the following results when the range of elements is unknown to the algorithm:
\begin{itemize}
    \item[-] For any $\eps \in [O(\log n/n), O(1)]$, we give a deterministic algorithm with a cost of $O(\log^3 n / \eps)$ for an array with $(1+\eps) n$ space. 

    \item[-] For any $\gamma \in [O(1), O(\log^2 n)]$, we give a deterministic algorithm with a cost of $O(\log^3 n / \gamma)$ for an array with $\gamma n$ space. 
\end{itemize}

\paragraph{Our Techniques.} We now give an overview of our main techniques. We describe these for the simpler case when the range of elements is exactly $[0,1]$. At the end of this discussion, we outline the additional ideas required to extend our results to the case when the range is unknown to the algorithm. Our main technical contribution is to introduce two new data structures {\em elementary trees} and {\em composite trees} for the online sorting problem. Before describing these data structures, we remark that these are {\em virtual} trees, i.e., while these help describe our online sorting algorithms for inserting elements into the array, there is no tree data structure being actually used to store elements. 

We start by describing elementary trees. An elementary tree is simply a complete binary tree whose leaves correspond to the array cells. We label each node of the elementary tree by a range that corresponds to the range of elements stored in the leaves underneath. Since the entire array accommodates all the elements, we start by labeling the root $[0, 1]$. Recursively, we would now like to allocate the left and right halves of the array to half-intervals $[0, 1/2]$ and $[1/2, 1]$. But, doing so statically might be wasteful in terms of space. E.g., if all the elements are in $[0, 1/2]$, then the right half-array would go unused. Instead, we do a dynamic allocation of the half-arrays to half-intervals. Suppose the initial element is in $[1/2, 1]$. Then, we label the left child of the root $[1/2, 1]$. Later, if we get an element from $[0, 1/2]$, we label the right child $[0, 1/2]$. We emphasize that a left child can be labeled by the right half of the parent's interval, and vice-versa. Moreover, it might also be the case that before getting an element from $[0, 1/2]$, we get an element from $[1/2, 1]$ that we cannot recursively insert in the left half-array. In this case, we will label the right child of the root $[1/2, 1]$ and insert the element in the right half-array. This added flexibility that allows both children to be labeled by the same half of the parent's interval dynamically adapts the space allocation to density of elements in different parts of the interval, thereby reducing wasted space. We continue this labeling scheme recursively to nodes further down in the tree, until we get to a leaf where the element is actually inserted. 

Our first result that uses $(1+\eps)n$ space only employs elementary trees. The idea is simple: we insert elements in a sequence of elementary trees, adding a new tree every time an element cannot be inserted into any of the previous trees. To obtain the space bound, we focus on the unused subarrays, which correspond to maximal empty subtrees of the elementary trees. We show the key property that any two subtrees of the same height (possibly in different elementary trees) that have empty child subtrees must have disjoint range labels (at their roots). This gives us a powerful tool to bound the empty space across all the elementary trees. By choosing the height of the trees carefully, we show that at most $\frac{\eps n}{\log n}$ space is left unused in subtrees of any given height, which yields an overall bound of $\eps n$ on the total unused space.

In composite trees, the goal is to reduce the cost of the solution compared to elementary trees by using more space in the data structure. Because the space allocation is entirely dynamic in an elementary tree, the elements can be completely out of sorted order based on the online arrival sequence. The other extreme is the $O(n^2)$-space data structure that we described earlier, where we statically allocate $n$ cells to each sub-interval of length $1/n$. This ensures that elements are essentially in sorted order (up to $1/n$), but all but one of these statically allocated sub-arrays might be left unused. A hybrid between these two extremes is to make a small static allocation to each sub-interval, and if any sub-interval has more elements than its allocated space, then use elementary trees to handle the overflow. A composite tree is essentially a recursive implementation of this hybrid structure. It is a complete binary tree, where a node at level $h$ is an elementary tree of height $h$. The interval labels of the composite tree nodes, i.e. those of root nodes in the corresponding elementary trees, are static, but those of non-root nodes in the elementary trees are dynamic. The algorithm first tries to insert an arriving element at the leaf node in the composite tree that is designated for the sub-interval containing the element. If this subarray is full, then the algorithm attempts to insert the element in the elementary tree at the parent of this leaf node. In this manner, the algorithm moves up from the leaf node of the composite tree along the path to the root until the element has been successfully inserted in an elementary tree on this path.

We use composite trees in two ways. For the data structure that uses between $O(n)$ and $O(n\log n)$ space, we use a single composite tree followed by a sequence of elementary trees, which we call a tree ensemble. The algorithm tries to insert each element in the composite tree, failing which it inserts the element greedily in the elementary trees. When we are allowed more space, namely between $O(n\log n)$ and $O(n\log^2 n)$ space, then we use an array of composite trees only. In this case, we statically segment the entire range $[0, 1)$ into smaller subranges, each of which is allocated a distinct composite tree. An arriving element is inserted in the unique composite tree allocated for the subrange containing the element. This strategy of static allocation helps reduce the cost of the allocation down to a constant, as desired, but the space allocation is larger since all but one of the composite trees might end up being empty.

Finally, we outline the additional ideas required for the extensions to elements with unknown range. It is natural to use a ``guess-and-double'' strategy. Every time an element is outside the current guessed range, the algorithm revises the guessed range such that the new element is included in the new range and it doubles the old range. Correspondingly, a new copy of the data structure for fixed range is added. This strategy has the advantage that the costs incurred by the algorithm form a geometric series and therefore, is comparable to the final guess. Recovering the space bounds, however, is more complicated. Na\"ively, the algorithm allocates a new copy of the entire data structure for every guessed range. It is possible that the guess is updated by every new element, thereby incurring a $\Theta(n)$ overhead in space. To alleviate this problem, we run an ``inner'' guess-and-double loop on the number of elements to be inserted in the data structure for the current range. Whenever the number of elements exceeds the current guess, we create a new copy of the data structure while doubling the guess on the number of elements. Since the doubling happens at most $O(\log n)$ times, this increases the cost by a factor of $O(\log n)$. The total space used across these iterations of the inner loop form a geometric series and is dominated by the final space allocation. This incurs a constant space overhead, which is sufficient for $\gamma\in [O(1), O(\log^2 n)]$. To obtain a $(1+\eps)n$-space data structure for unknown range, we need one more idea since this data structure cannot tolerate constant overhead in space. Instead of using the fixed range data structure in a black box manner, we actually use its properties for this extension. In particular, we show that it allocates space in small enough quanta that the overall unused space across the different iterations of the guess-and-double loops can be bounded by $\eps n$. This gives a total space bound of $(1+\eps)n$.

\paragraph{Pre-processing.} We now give a simple pre-processing step that will simplify the recursive description of our algorithms. Recall that all elements are in the range $[0, 1]$. If the value of an element is precisely $1$, we change it to $1-1/n$ before applying the online sorting algorithm on it. Note that this only changes the value of the element by at most $1/n$; across all the $n$ elements, this adds a cost of at most $2$ to the online sorting solution. The advantage is that we can now assume that all elements are in the half-open interval $[0, 1)$, which can be recursively split into disjoint half-open intervals (the latter helps describe our algorithm cleanly). The same pre-processing is valid when the elements have unknown range, except that we first scale the current guess on the range down to $[0, 1]$ before applying this transformation.

\paragraph{Related Work.}
As mentioned earlier, the online sorting problem was introduced by Aamand {\em et al.}~\cite{AamandA0K23}, who gave deterministic upper and lower bounds of $c= O(\sqrt{n})$ for $m=n$, a deterministic lower bound of $\gamma\cdot c = \Omega(\log n / \log \log n)$, and a deterministic upper bound of $c = 2^{\sqrt{\log n}\sqrt{\log\log n + \log(1/\eps)}}$ for $m = (1+\eps)n$. The lower bound of $\Omega(\sqrt{n})$ for $m = n$ was later extended to randomized algorithms~\cite{AbrahamsenB0K024}. Better results are known if the elements are drawn independently and uniformly at random from $[0, 1]$. In this case, Abrahamsen {\em et al.}~\cite{AbrahamsenB0K024} showed an upper bound of $O((n \log n)^{1/4})$ with $m = n$ and of $O(1+1/\eps)$ with $m = (1+\eps) n$. The upper bound for the case \( m = n \) was recently improved by Hu~\cite{improved-stochastic-hu} to \( \log n \cdot 2^{O(\log^* n)} \) and independently by Fotakis {\em et al.}~\cite{FotakisKPT25} to $O(\log^2 n)$. To complement these upper bounds, Hu~\cite{improved-stochastic-hu} also gave 
an almost matching lower bound of \( \Omega(\log n) \) for the stochastic problem.

The online sorting problem was originally conceived as a tool to better understand online packing problems for convex shapes. This latter class includes many well-studied problems such as online strip packing~\cite{BakerS83,CsirikW97,HanIYZ07,YeHZ09}, online rectangle bin packing~\cite{CoppersmithR89,Galambos91,Vliet92,GalambosV94,CsirikFL93,CsirikV93,SeidenS03,HanCTZZ11}, etc. In a different direction, the online sorting problem can be viewed as designing a minimum length tour on points arriving online on the line metric in the interval $[0, 1]$. This view naturally generalizes to other metric spaces, where it defines an online version of the traveling salesman problem (TSP). This problem was studied by Abrahamsen {\em et al.}~\cite{AbrahamsenB0K024} for the uniform metric, where they showed tight upper and lower bounds of $\Theta(\log n)$ on the competitive ratio. They also posed the question for general metrics, which was settled very recently by Bertram who showed an upper bound of $O(\sqrt{n})$ for all metric spaces~\cite{bertram2025onlinemetrictsp}. Given the lower bound on the line metric, this result is asymptotically tight. 

\paragraph{Independent and Concurrent Work.}
We have been made aware of (unpublished) independent and concurrent work by Nirjhor and Wein~\cite{NirjhorW25} on the same problem that we study in this paper. They obtain a cost bound of $(\log n)^{O(\log \log n + \log (1/\eps))}$ with $(1+\eps)n$-space. In contrast, we obtain a cost bound of $O(\log^2 n / \eps)$ for the same space bound in this paper.

\paragraph{Organization.} \Cref{sec:elementary} introduces one of our two main data structures, an {\em elementary tree}. In \Cref{sec:small-space}, we use elementary trees to define the online sorting algorithm with $m = (1+\eps)n$ space. We introduce our second main data structures, {\em composite trees}, in \Cref{sec:composite}. We then use elementary and composite trees to define {\em tree ensembles} that are used in our online sorting algorithm for $m \in [O(n), O(n\log^2 n)]$ space in \Cref{sec:small-error-part1,sec:small-error-part2}. We give the extensions to the setting with unknown range of elements in \Cref{sec:doubling}. We conclude with some closing remarks in \Cref{sec:closing}.

\section{Core Primitive: Elementary Trees}\label{sec:elementary}The data structures that we propose for the online sorting problem are based on two new primitives, {\em elementary trees} and {\em composite trees}. We introduce elementary trees in this section, and composite trees in \Cref{sec:composite}.


     An {\em elementary tree} of height $h$ is an array data structure of size $2^h$ that has an associated complete binary tree of height $h$. We emphasize that the tree is {\em virtual} in the sense that it will be used for defining how elements are inserted in the array, but the actual data structure is just an array and not a tree. The leaves of the tree, at height $0$, are each associated with a unique array cell, enumerated left to right as shown in Fig.~\ref{fig:elementary-tree}. So, all insertions of elements are at the leaves of the tree. The root of the tree is at height $h$, its children are at height $h-1$, their children at height $h-2$, and so on.

\begin{figure}[tbh]
\centering
\includegraphics[width=0.8\linewidth]{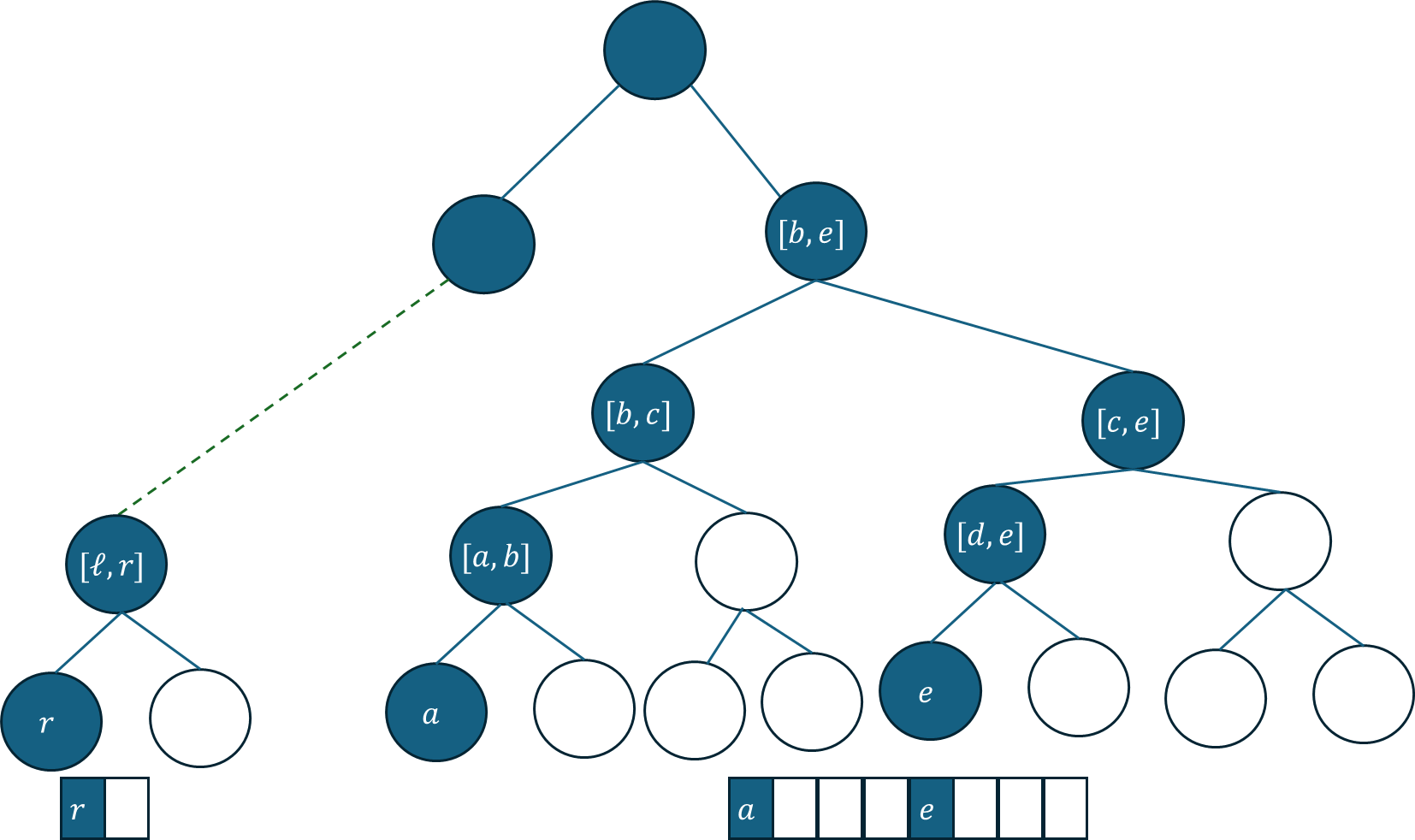}
    \caption{Illustration of an elementary tree. The colored cells in the array represent occupied cells. Similarly, the colored nodes in the tree represent marked nodes.}
\label{fig:elementary-tree}    
\end{figure}

\begin{figure}[tbh]
    \centering
    \subfloat[\centering Before the insertion]{{\includegraphics[width=.4\linewidth]{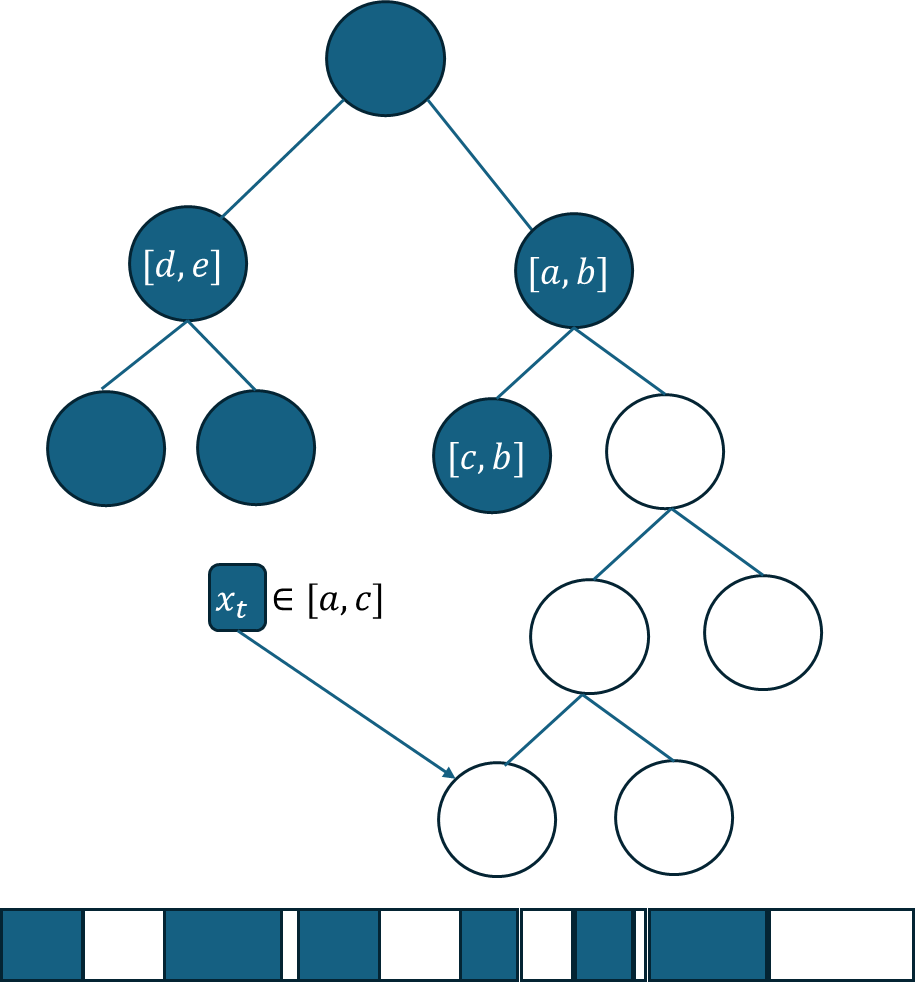} }}%
    \hfill
    \subfloat[\centering After the insertion]{{\includegraphics[width=.4\linewidth]{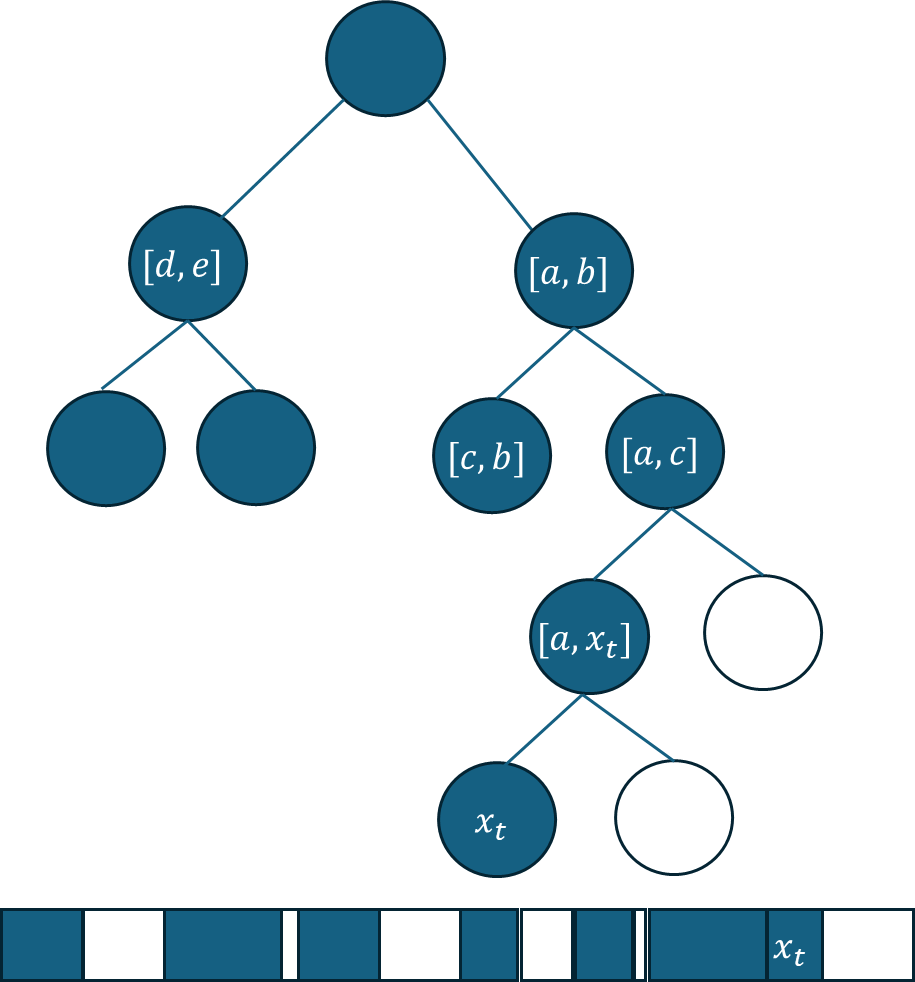} }}%
    \caption{Insertion of an element $x_t$ in an elementary tree. The filled parts of the array represent occupied regions. Similarly, the colored nodes in the tree represent marked nodes.}%
    \label{fig:insertion to elementary}
\end{figure}
     \paragraph{Insertion Algorithm and Node Labels.} 
     We now describe the online algorithm that inserts elements in an elementary tree data structure. In addition to inserting elements, the algorithm also labels nodes of the elementary tree with sub-intervals of $[0, 1)$. We say that a node is {\em marked} if it has already been labeled, and {\em unmarked} otherwise. The algorithm maintains the invariant that the label of a marked child node is either the left half or the right half of the label of its parent node. (A node cannot be marked before its parent.) Note that the two children of a node can be marked with the same label or different ones, and one child or both can be unmarked. For a marked node $v$, we say that an element $x$ is {\em admissible} at $v$ if the value $x$ is contained in the label of $v$; otherwise, $x$ is {\em inadmissible} at $v$. 
 
     We assume that the root node is labeled with an interval that includes all elements that will be inserted into the elementary tree. Therefore, the root node is marked and every element that the algorithm tries to insert into the elementary tree is admissible at the root. Let $\ins(v, x)$ denote the (recursive) algorithm that tries to insert $x$ in the subtree rooted at a node $v$. (So, we call $\ins(r, x)$ for the root node $r$ when we insert $x$ into the elementary tree.) The procedure $\ins(v, x)$ either terminates with ``success'' if $x$ is inserted at a leaf of the subtree rooted at $v$, or ``failure'' if it fails to do so. Eventually, $\ins(r, x)$ also terminates with either success or failure, corresponding to $x$ being inserted in the elementary tree or failing to do so, respectively. 
     
     If $v$ is not a leaf node, then let $v_\ell$ and $v_r$ respectively denote the left and right child of $v$. We now define $\ins(v, x)$ as follows:
     \begin{itemize}
        \item[-] If $v$ is a leaf node and is unmarked, then label $v$ with the left or right half of its parent's label that contains $x$, write $x$ in the array cell at the leaf node $v$, and return success.
        \item[-] If $v$ is marked and $x$ is inadmissible at $v$, then return failure.
        \item[-] If $v$ is marked and $x$ is admissible at $v$, then call $\ins(v_\ell,x)$. If it returns success, return success. Otherwise, if it returns failure, then call $\ins(v_r,x)$. If this returns success, then return success. Otherwise, if it  returns failure as well, then return failure.
        \item[-] If $v$ is unmarked, then label $v$ with the left or right half of its parent's label that contains $x$ and call $\ins(v_\ell,x)$.
    \end{itemize}
     Note that we are essentially performing an in-order traversal of the elementary tree with the following invariant: 
     \begin{claim}\label{cl:label}
        If an element $x$ is written in the array cell at a leaf node $v$ of an elementary tree, then $x$ is admissible at every  ancestor of $v$.
    \end{claim}
    \begin{proof}
         Let us denote the vertices on the path from $v$ to the root $r$ as $v_0 = v, v_1, v_2, \ldots, v_h = r$. Note that $\ins(v, x)$ must have been called to write $x$ in the array cell at $v$. The only procedure that can call $\ins(v, x)$ is $\ins(v_1, x)$, which in turn could have only been called by $\ins(v_2, x)$, and so on. Moreover, $\ins(v_{i+1}, x)$ calls $\ins(v_i, x)$ only if $x$ is admissible at $v_{i+1}$. This establishes the lemma.
    \end{proof}

    The following lemma bounds the cost in an elementary tree:
    \begin{lemma}[Elementary Tree Cost Lemma]\label{lem:elementary-cost}
        Suppose a set of elements is inserted in an elementary tree $A$ of height $h$. Let $\ell$ denote the length of the interval label of the root of $A$. Then, the total cost of the elements in $A$ is at most $\ell h$.
    \end{lemma}
\begin{proof}
    Let $x_i, x_{i+1}$ be a pair of elements that occupy consecutive non-empty cells in the array. Let $v$ denote the least common ancestor of $x_i, x_{i+1}$ in $A$, and let $v_\ell, v_r$ be the left and right child of $v$. Then, $x_i$ is the rightmost non-empty leaf in the subtree under $v_\ell$ and $x_{i+1}$ is the leftmost non-empty leaf in the subtree under $v_r$. By \Cref{cl:label}, both $x_i, x_{i+1}$ are contained in the interval label of $v$. Thus, their absolute difference is at most the length of the interval label of $v$. We charge $|x_i - x_{i+1}|$ to $v$. Every non-leaf node in the elementary tree is charged at most once, and the charge is at most the length of its interval label. The lemma follows since the root has an interval label of length $\ell$, its  2 children have interval labels of length $\ell/2$ each, their 4 children have interval labels of length $\ell/4$ each, and so on.
\end{proof}

\section{A $(1+\eps) n$-space, $O\left(\frac{\log^2 n}{\eps}\right)$-cost Data Structure}\label{sec:small-space}

\subsection{Description of the Data Structure}
The data structure comprises a series of identical elementary trees $A_0, A_1, \ldots$, each of height 
\[
    H := \lfloor\log (\eps n) - \log \log n\rfloor\footnote{All logarithms are with base $2$ unless otherwise mentioned.}  
\]    
and with the roots labeled $[0,1)$.
The number of elementary trees will be calculated indirectly via the space bound.

\begin{remark}
    \label{rem:epsilon-cannot-be-too-small}
    Here, we use the fact that $\eps \ge \log n / n$, since the height $H$ of the elementary trees must be non-negative.
\end{remark}

\subsection{The Online Sorting Algorithm}

We now give the online sorting algorithm, i.e., the algorithm that inserts elements into the data structure described above. Let $x_t$ be the element being inserted. We iterate over the arrays in order. I.e., we try to insert $x_t$ in $A_0$, then $A_1$, then $A_2$, etc., until the element is successfully inserted. Since we have not fixed the number of elementary trees in advance, the algorithm never fails to insert an element. The algorithm for inserting an element $x_t$ in a specific elementary tree $A_k$ is as given in \Cref{sec:elementary}.

\subsection{Space Bound}

Note that we did not specify the number of elementary trees in the description of the data structure. Here, we bound the number of elementary trees, and hence the total space used. To do so, we first categorize nodes in the elementary trees as {\em full}, {\em partial}, or {\em empty}. A node $v$ is empty if it is unmarked at the end of the algorithm. (Although root nodes are marked at the outset, we call a root node empty even if it is marked as long as the tree does not contain any element.)
A non-leaf node is full if neither of its children is empty, and partial if exactly one child is empty. (If both children are empty, then the parent node must also be empty. This is because if it is a non-root node, then it cannot be marked and is therefore empty, while if it is a root node with two empty children, then it is empty by definition.) A leaf node is said to be full if it is not empty. 

Now that we have defined full, partial, and empty nodes, we can state the main property of the online sorting algorithm that results in the space bound.
\begin{lemma}\label{lem:small-space}
    Fix any height $h$. Then, all partial nodes at height $h$ are labeled by disjoint intervals.
\end{lemma}
To prove this lemma, we first prove some auxiliary claims about the algorithm.

\begin{claim}\label{cl:elem-greedy-one}
    If the online algorithm calls $\ins(v, x)$ and $v$ is empty, then $x$ will be successfully inserted at some leaf node in the subtree under $v$.
\end{claim}
\begin{proof}
    In this case, the algorithm follows the path given by the left child at each node and inserts element $x$ at the leftmost leaf in the subtree under $v$.
\end{proof}

\begin{claim}\label{cl:elem-greedy-two}
    Suppose the online algorithm calls $\ins(v, x)$ during the insertion of element $x$ for some node $v$ in an elementary tree. If there is an empty node $w$ in the subtree under $v$ and the parent of $w$ (call it $u$) is admissible for $x$, then $x$ will be successfully inserted at some leaf node in the subtree under $v$.
\end{claim}
\begin{proof}
    We prove the claim by induction on the length of the $u-v$ path. First, consider the base case $u = v$. In this case, one of the children of $v$ is $w$, which is an empty node. Let $w'$ be the other child of $v$. If $\ins(v, x)$ calls $\ins(w, x)$ before $\ins(w', x)$, then $x$ is inserted in the subtree under $w$ by \Cref{cl:elem-greedy-one}. Otherwise, $\ins(v, x)$ calls $\ins(w', x)$ before $\ins(w, x)$. If $\ins(w', x)$ terminates with success, then $x$ is inserted in the subtree under $w'$. Otherwise, $\ins(w', x)$ terminates with failure and $\ins(w, x)$ is called, which inserts $x$ in the subtree under $w$ by \Cref{cl:elem-greedy-one}. In either case, $x$ is inserted in the subtree under $v$.
    
    Now, consider the inductive case $u \not= v$. Let $v'$ be the child of $v$ that is an ancestor of $u$, and let $v''$ be the other child. If $\ins(v, x)$ calls $\ins(v', x)$ before $\ins(v'', x)$, then $x$  is inserted in the subtree under $v'$ by the inductive hypothesis. Otherwise, $\ins(v, x)$ calls $\ins(v'', x)$ before $\ins(v', x)$. If $\ins(v'', x)$ returns success, then $x$ is inserted in the subtree under $v''$. If $\ins(v'', x)$ returns failure, then $\ins(v', x)$ is called and by the inductive hypothesis, $x$ is inserted in the subtree under $v'$. In either case, $x$ is inserted in the subtree under $v$.
\end{proof} 

We use these claims to show \Cref{lem:small-space}:

\begin{proof}[Proof of \Cref{lem:small-space}]
    First, observe that any two intervals marking nodes at the same height have the same length, and the node labels form a laminar set. Therefore, any two nodes at the same height are either marked by the same interval or by completely disjoint intervals. Now, suppose two partial nodes $v, v'$ at height $h$ are marked by the same interval $I$ during the insertion of elements $x_t, x_{t'}$ respectively, for $t' \ge t$. First, we rule out $t' = t$. By \Cref{cl:elem-greedy-one}, if the online algorithm marks a node when inserting an element, then the element is inserted successfully at a leaf in the subtree under the newly marked node. Since $v, v'$ are at the same height, the subtrees under them are disjoint. Therefore, it must be that $x_t \not= x_{t'}$. Therefore, $t \not= t'$, i.e. it must be that $t' > t$.
    
    Let $u, u'$ be the highest ancestors of $v, v'$ that are not ancestors of both $v, v'$. (If $v, v'$ are in the same elementary tree, then $u, u'$ are children of their least common ancestor. If $v, v'$ are in different elementary trees, then $u, u'$ are the roots of their respective trees.) By \Cref{cl:label}, all ancestors of $v, v'$ are admissible for both $x_t, x_{t'}$. Now, let $w'$ be the deepest ancestor of $v'$ that is marked before the insertion of $x_t$. Since the entire path from $u'$ to $w'$ is admissible for $x_t$ and one of the children of $w'$ was empty, by \Cref{cl:elem-greedy-two}, the element $x_t$ would have been successfully inserted in the subtree under $u'$ if $\ins(u', x_t)$ had been called. This means that $\ins(u, x_t)$ was called before $\ins(u', x_t)$, and $\ins(u, x_t)$ inserted $x_t$ successfully in the subtree under $u$, thereby not requiring the algorithm to call $\ins(u', x_t)$. Subsequently, at time $t'$, the algorithm will maintain the same order, i.e., it calls $\ins(u, x_{t'})$ first, and $\ins(u', x_{t'})$ is called only if $\ins(u, x_{t'})$ returns failure. But, $\ins(u, x_{t'})$ cannot return failure. This is by \Cref{cl:elem-greedy-two}, because the entire path from $u$ to $v$ is admissible for $x_{t'}$, and  $v$ has an empty child node. Thus, $\ins(u, x_{t'})$ will successfully insert $x_{t'}$ in the subtree under $u$. This means node $v'$ cannot be marked during the insertion of $x_{t'}$, a contradiction.
\end{proof}
As a simple corollary of this lemma, we obtain the following property:
\begin{lemma}\label{lem:smallspace-unused-space}
    The unused space in the data structure above is at most $\eps n$ at any time. As a consequence, the data structure uses at most $(1+\eps)n$ space after inserting all $n$ elements.
\end{lemma}
\begin{proof}
    We set up a charging scheme that bounds the unused space in the data structure. Every maximal empty subtree of height $h < H$ (comprising $2^h$ empty array cells) is a child subtree of a partial node at height $h+1$. We charge the empty space of $2^h$ to this partial node. By \Cref{lem:small-space}, the partial nodes at height $h+1$ must be labeled by distinct intervals. The label of a node at height $h$ is an interval of length $2^{h-H}$. Putting these together, we can conclude that the total charge on the nodes at height $h+1$ is at most $2^H$. Summing over heights $h = 0, 1, \ldots, H-1$ and since $\eps$ is at most a constant, the total empty space is at most 
\[
    H\cdot 2^H
    \le \frac{\eps n}{\log{n}} \cdot  \lfloor\log({\eps n) - \log{\log{n }}}\rfloor
   \le \eps n. \qedhere
\]    
\end{proof}

\eat{
The space bound follows immediately from the above lemma:
\begin{lemma}\label{lem:smallspace-space}
    The total space used in the data structure above is at most $(1+\eps)n$.
\end{lemma}
\begin{proof}
By \cref{lem:smallspace-unused-space} the total amount of unused space is $\epsilon n$,
since the used space is at most $n$, the total space is at most $(1+\eps)n$.
\end{proof}
}

\eat{
To this end, we define an empty node to be {\em maximal} if its parent is partial (and not empty itself). Let $\cE(h)$ denote the set of maximal empty nodes at level $h$ across all the virtual trees. The subtree under an empty node at height $h$ contains $2^{h}$ empty cells. Therefore, the total number of empty cells across all the arrays is given by $\sum_{h=0}^H |\cE(h)|\cdot 2^{h}$. 

Note that the label of any node at height $h$ is an interval of length $2^{h-H}$. Therefore, by \Cref{lem:space-main}, we have $|\cE(h)| \le 2^{H-h}$. Therefore, the total number of empty cells across all the arrays is at most 
\begin{align*}
    &\sum_{h=0}^H |\cE(h)|\cdot 2^{h} 
   \le \sum_{h=0}^H 2^{H-h} \cdot 2^{h} 
   = \sum_{h=0}^H 2^{H}
   = 2^{H}\cdot (H+1) \\
   & \le \lfloor\frac{\eps n}{\log{n}} \rfloor \cdot  \lfloor\log({\eps n) - \log{\log{n }}}+1\rfloor
   \le \lfloor\frac{\eps n}{\log{n}} \rfloor \cdot  \lfloor\log{n} \rfloor \le 
   \lfloor \eps n \rfloor.
\end{align*}
The second to last inequality follows by assuming that $n\ge 4$.
}

\subsection{Cost Bound}
Note that by \cref{lem:elementary-cost}, the cost incurred by a single elementary tree is at most 
\begin{equation}\label{eq:single}
    H = \lfloor\log (\eps n) - \log \log n\rfloor \le \log (\eps n),
\end{equation}    
since the roots are labeled $[0, 1)$.
We use this property to bound the total cost.
\begin{lemma}\label{lem:cost-smallspace}
    The total cost of the data structure above is $O\left(\frac{\log^2 n}{\eps}\right)$.
\end{lemma}
\begin{proof}
    By \Cref{lem:smallspace-unused-space}, the total space is at most $(1+\eps) n$, while that of a single elementary tree is
    \[
        2^H = 2^{\lfloor\log (\eps n) - \log \log n \rfloor} \geq 2^{\log (\eps n) - \log \log n - 1} = \frac{\eps n}{2 \log n}.
    \]    
    Therefore, the number of elementary trees is at most 
    \[
        \frac{(1+\eps)n}{\frac{\eps n}{2\log n}} = 2\left(1+1 / \eps\right) \, \log n.
    \]
     By \Cref{eq:single}, the cumulative cost is at most
     \[
        2\left(1+ 1 / \eps\right) \, \log n \cdot \log (\eps 
 n) = O\left(\frac{\log^2 n}{\eps}\right).
 \]
 
     We are only left with the cost incurred {\em between} elementary trees, i.e., by the rightmost element in $A_k$ and the leftmost one in $A_{k+1}$. This cost is at most 1 for every pair of consecutive elementary trees, and hence, $O\left(\frac{\log n}{\eps}\right)$ overall.     
     The lemma follows by adding the two types of cost.
\end{proof}

\section{Core Primitive: Composite Trees}\label{sec:composite}    We now define our second core primitive -- a data structure that we call a {\em composite tree} -- by using elementary trees as building blocks. 
    
    A composite tree of height $h$ is an array data structure that has an associated complete binary tree, where each node of the binary tree is an elementary tree. We again emphasize that just like elementary trees, a composite tree is {\em virtual} in that it is used to define how elements are inserted in the array, but the actual data structure is just an array and not a tree. The leaves of a composite trees are at height $1$ (not $0$) and the root is at height $h$. A node at height $h'$ in a composite tree is an elementary tree of height $h'$, i.e., it is an array of size $2^{h'}$. We show a composite tree in Fig.~\ref{fig:composite-tree}. Note that we store the composite tree ``in order'' in the underlying array, i.e., the array corresponding to the left child subtree of the root (which is a composite tree of height $h-1$, recursively), followed by the array corresponding to the root node (which is an elementary tree of height $h$), followed by the array corresponding to the right child subtree of the root (which is a composite tree of height $h-1$, recursively).


\begin{figure}[tbh]
\caption{Illustration of a composite tree}
\label{fig:composite-tree}
\centering
\includegraphics[width=0.8\linewidth]{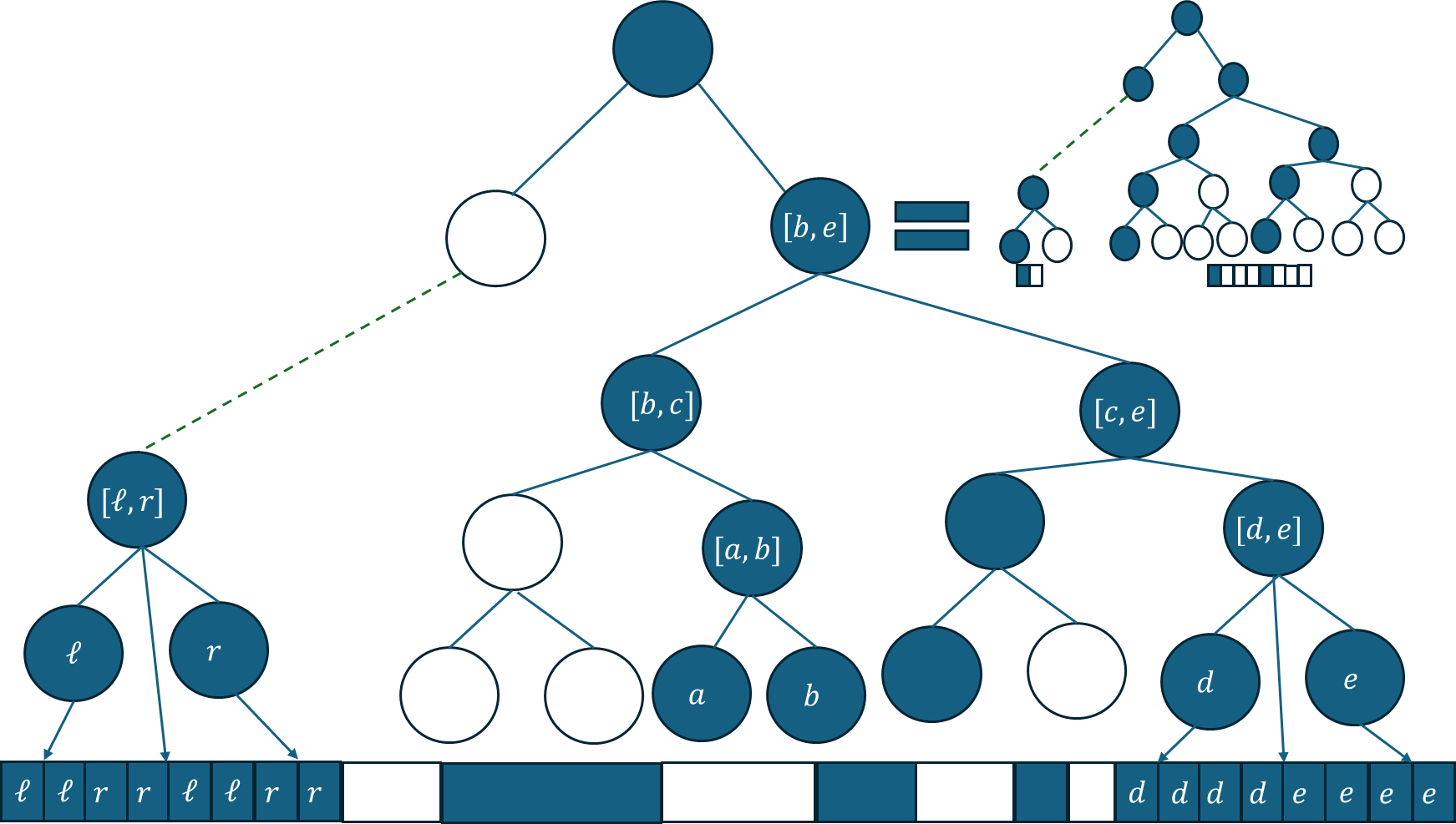}
\end{figure}

    \paragraph{Size of the Array.}
    We compute the size of the array in a composite tree:
    \begin{lemma}\label{lem:space-composite}
        The size of the array in a composite tree of height $h$ is $h\cdot 2^h$.
    \end{lemma}
    \begin{proof}
    Recursively, a composite tree of height $h$ comprises two composite trees of height $h-1$ each and an elementary tree of height $h$. The latter is an array of size $2^h$. If $s(h)$ denotes the size of a composite tree, we get $s(h) = 2\cdot s(h-1) + 2^h$. Furthermore, for $h=1$, a composite tree of height $1$ is an elementary tree of height $1$, which is an array of size $2$. Putting these together, we get that the size of a composite tree of height $h$ is $h\cdot 2^h$. 
    \end{proof}

    \paragraph{Node Labels.}
    Recall that the nodes of an elementary tree are labeled by sub-intervals of $[0, 1)$, where the root gets a static label before any element is inserted and the non-root nodes get labels decided dynamically by the online algorithm. The nodes of a composite tree are also labeled by intervals. Each node of a composite tree is actually an elementary tree; the label of a node in the composite tree is the same as the label of the root of the elementary tree represented by it. This means that the nodes of a composite tree are also labeled by sub-intervals of $[0, 1)$, but these labels are static and decided before elements are inserted. The labels of composite tree nodes satisfy the following property: the labels for $v_\ell$ and $v_r$, the left and right child of a node $v$, are the left and right half-intervals of the label for $v$. So, if the label of $v$ is $[a, a+\Delta)$, then that of $v_\ell$ and $v_r$ are $[a, a + \Delta/2)$ and $[a + \Delta/2, a + \Delta)$ respectively.

    \paragraph{Insertion Algorithm.}
    We now describe the online algorithm that tries to insert an element $x$ in a composite tree. Borrowing notation from elementary trees, we say that $x$ is admissible at a composite tree node if the interval label of the node contains $x$. (Note that unlike elementary trees, admissibility in a composite tree node can be decided before any insertions have happened, since the labels are static.)  

    To insert $x$, we first identify the unique leaf node $v$ in the composite tree where $x$ is admissible. The uniqueness follows from the fact that the left and right child nodes in a composite tree split the parent's label into disjoint half-intervals. Therefore, the interval labels of all the leaves are disjoint. Denote the path from the leaf node $v$ to the root node $r$ of the composite tree as $v_1 = v, v_2, \ldots, v_h = r$. 

    We now give the algorithm for inserting $x$ in the composite tree. Recall that $\ins(v_i, x)$ denotes the online algorithm that tries to insert $x$ into the elementary tree at node $v_i$. The online algorithm for a composite tree uses $\ins(v, x)$ as a subroutine:
    \begin{itemize}    
    \item[-] Initialize $i=1$.
    \item[-] Run $\ins(v_i, x)$. If $\ins(v_i, x)$ returns success, i.e., $x$ was inserted in the elementary tree at node $v_i$, then return success and exit the algorithm. 
    \item[-] Otherwise, $\ins(v_i, x)$ returns failure, i.e., it failed to insert $x$ in the elementary tree at node $v_i$. If $v_i$ is the root node, then return failure and exit the algorithm; else, increase $i$ by $1$ and go to the previous step.
    \end{itemize}

\subsection{Tree Ensembles}


We now consider an aggregate data structure comprising a single composite tree $A$ followed by $\alpha \ge 0$ elementary trees $B_1, B_2, \ldots, B_\alpha$. The composite tree and all the elementary trees are of the same height $h$. Furthermore, the interval label $I$ of the root of the composite tree $A$ is identical to that of the roots of the elementary trees $B_i$, $i = 1, 2, \ldots, \alpha$. 
We assume that every element that is being inserted in this data structure is in this interval $I$. On the arrival of an element $x$, the online sorting algorithm (call it $\cal O$) first tries to insert $x$ in the composite tree $A$ using the algorithm given earlier in this section. If this succeeds, $\cal O$ terminates with success. If this fails, then $\cal O$ tries to insert $x$ in $B_1, B_2, \ldots, B_\alpha$ in this order, using the algorithm given in \Cref{sec:elementary}. The first time that it succeeds to insert $x$ in some $B_i$, algorithm $\cal O$ terminates with success. If $\cal O$ fails to insert $x$ in any of $B_1, B_2 \ldots, B_\alpha$, then it terminates with failure. We call this data structure comprising $A, B_1, B_2, \ldots, B_\alpha$ a {\em tree ensemble} of order $\alpha$.

The main property of tree ensembles that enables its space efficiency is that the number of the elementary tree nodes labeled by an interval corresponds to the number of the elements in that interval inserted in the ensemble. We call this the {\em Space Lemma}.

\begin{lemma}[Space Lemma]\label{lem:main-space}
    Consider a sequence of elements being inserted in a tree ensemble $A, B_1, B_2, \ldots, B_\alpha$.  Suppose $r+1$ elementary tree nodes at height $h\ge 1$ are labeled by the same interval $I$. (This can include both the elementary trees corresponding to nodes of $A$ at height $\ge h$ and the elementary trees $B_1, B_2, \ldots, B_\alpha$.)
    Then, at least $r\cdot 2^h + 1$ elements in $I$ have been inserted into the tree ensemble.
\end{lemma}
\begin{proof}
    We prove this by induction on $h$. First, we consider the base case, i.e., elementary tree nodes at $h = 1$. If $r+1$ nodes are labeled by the same interval $I$, then for the first $r$ of these nodes, each of their two children leaf nodes must already be occupied by some element in $I$. Moreover, for the rightmost node labeled by $I$, at least one of its children leaf nodes must be occupied. Therefore, there are at least $2r+1$ elements in $I$ that have already been inserted in the corresponding arrays, and hence, the lemma holds for $h=1$.

    We now prove the inductive case. 
    Suppose $r+1$ elementary tree nodes at height $h\ge 2$ are labeled by an interval $I$. For the first $r$ of these nodes, both children must also be labeled, else the $(r+1)$st node would not have been labeled by $I$. Furthermore, for the rightmost node labeled $I$, at least one its children must be labeled. Finally, note that there are two nodes in $A$ at height $h-1$ labeled by the left and right half-intervals of $I$, and these are elementary trees of height $h-1$. Therefore, the total number of nodes at height $h-1$ that are labeled by either the left or right half-interval of $I$ is at least $2r+3$. 
    
    Note that there is at least one node labeled by each of the left and right half-intervals. Let $\lambda + 1$ nodes be labeled with the left half-interval of $I$ and $\rho + 1$ nodes be labeled with the right half-interval of $I$. Thus, $\lambda + \rho \ge 2r+1$. Applying the inductive hypothesis on these labeled nodes at height $h-1$, we get that the number of elements in $I$ that have been inserted in the tree ensemble is at least 
    \[
        \left(\lambda\cdot  2^{h-1} + 1\right) + \left(\rho\cdot 2^{h-1} + 1\right)
        \ge (2r+1)\ 2^{h-1} + 2
        > r \cdot 2^h + 1. \qedhere
    \]
\eat{
    Next, suppose $\lambda = 0$. (The case of $\rho = 0$ is symmetric, so we skip it for brevity.) In this case, $\rho \ge 2r+2$. Applying the inductive hypothesis on these labeled nodes at height $h-1$, we get that the number of elements in (the right sub-interval of) $I$ that have been inserted in the array is at least 
    \[
        \left((\rho-1) 2^{h-1} + 1\right)
        = (2r+1)\cdot 2^{h-1} + 1
        > r \cdot 2^h + 1. \qedhere
    \]    
}    
\end{proof}

{\bf Remark:} Let $I^*$ be the interval that labels the roots of $A, B_1, B_2, \ldots, B_\alpha$. The lemma applies as stated to any sub-interval of $I^*$. It also applies to the interval $I^*$ itself, but since the roots are statically labeled by $I^*$ before any element is inserted, we define $r+1$ in this case as the number of \underline{non-empty} elementary trees, including the elementary tree at the root of $A$ and $B_1, B_2, \ldots, B_\alpha$. 

The following is a useful corollary of the space lemma when applied only to a composite tree.
\begin{corollary}\label{cor:composite-space}
    Consider a sequence of elements being inserted in a composite tree $A$. If the elementary tree representing a node $v$ in $A$ at height $h$ and with interval label $I$ is non-empty, then at least $2^{h-1}+1$ elements in $I$ have been inserted into $A$.
\end{corollary}
\begin{proof}
    Let $u$ be the root of the elementary tree for $v$, and let $u_\ell, u_r$ be its left and right child. At least one of the subtrees under $u_\ell, u_r$ is non-empty; thus, by \Cref{cl:label}, at least one of $u_\ell, u_r$ is labeled by a half-interval of $I$; call this half-interval $I'$. Note that there is a node in $A$ at height $h-1$ that also has label $I'$, and so does the root of the elementary tree at this node. Therefore, there are at least $2$ elementary tree nodes labeled by $I'$. It follows by \Cref{lem:main-space} that at least $2^{h-1}$ elements in $I'$ (that are also in $I$) have been inserted in $A$. 
\end{proof}

Next, we consider the cost of a solution produced by a tree ensemble. Let $x_i, x_{i+1}$ be two elements that occupy consecutive non-empty array cells in the tree ensemble. We define three types of costs:
\begin{itemize}
    \item[-] \one cost:  If $x_i, x_{i+1}$ are in the same elementary tree, either corresponding to a node in $A$ or in one of $B_1, B_2, \ldots, B_\alpha$, then we say that the cost between $x_i, x_{i+1}$ is of \one.    
    \item[-] \two cost: If $x_i, x_{i+1}$ are in elementary trees corresponding to different nodes $u, v$ in $A$, then we say that the cost between $x_i, x_{i+1}$ is of \two.
    \item[-] \three cost: If $x_i$ is in $A$ and $x_2$ is in $B_1$, or if $x_i$ is in $B_i$ and $x_{i+1}$ is in $B_{i+1}$ for some $i \ge 1$, then we say that the cost between $x_i, x_{i+1}$ is of \three.
\end{itemize}

The next lemma bounds each type of cost in a tree ensemble.

\begin{lemma}[Cost Lemma]\label{lem:main-cost}
     Suppose a set of $n$ elements is inserted in a tree ensemble denoted $A, B_0, B_1, \ldots, B_{\alpha-1}$. Let $\ell, H$ respectively denote the length of the interval label of the roots and the height of the trees. Then, the cost of the elements in the ensemble is given by:
     \begin{itemize}
        \item[-] the total \one cost is at most $n H \cdot (\ell/2^H)$.
        \item[-] the total \two cost is at most $9nH \cdot (\ell/2^H) + \ell$.
        \item[-] The total \three cost is at most $\alpha \ell$.
     \end{itemize}
\end{lemma}
\begin{proof}
    First, we consider \one cost. Let $v$ be the least common ancestor of $x_i, x_{i+1}$ in the elementary tree containing them, and let $v_\ell, v_r$ be the left and right child node of $v$. Then, $x_i$ is the rightmost non-empty leaf in the subtree under $v_\ell$ and $x_{i+1}$ is the leftmost non-empty leaf in the subtree under $v_r$. By \Cref{cl:label}, both $x_i, x_{i+1}$ are in the interval label of $v$. Next, consider the interval labels of $v_\ell, v_r$; call these $I_\ell, I_r$. If $I_\ell = I_r = I$ (say), then by \Cref{cl:label}, we have that $x_i, x_{i+1}$ are both in $I$ and $|x_i - x_{i+1}|$ is at most the length of $I$. In this case, we give a charge to $v_\ell$ equal to the length of $I$; $v_r$ receives no charge. The other case is that $v_\ell, v_r$ have different interval labels, i.e., $I_\ell not= I_r$. In this case, $|x_i - x_{i+1}|$ is at most the sum of lengths of the intervals $I_\ell$ and $I_r$, which is the length of the interval label of $v$. In this case, we give both $v_\ell, v_r$ a charge equal to the length of their interval labels $I_\ell, I_r$. 
    
    In the charging process above, every node is charged at most the length of its interval label. Furthermore, if a node is the right child of its parent and has the same interval label as its left sibling, then the node is not charged. First, we consider the charge on non-leaf nodes of elementary trees. Let $I$ be an interval label of such a node that was charged. Since this is a non-leaf node, it's height is some $h \ge 1$. The key observation is that there is at least one elementary tree node in the ensemble that is labeled by $I$ and did not get charged. In particular, consider the composite tree node at height $h$ labeled by $I$. In the elementary tree representing this node, the root is labeled $I$. But, this root could not have been charged in the above process since the nodes that were charged ($v_\ell, v_r$) were children of some node $v$, and therefore, no root node was charged. It follows that if $r$ nodes labeled by $I$ are charged overall, there are at least $r+1$ nodes labeled by $I$ in the ensemble. By \Cref{lem:main-space}, this implies that there are at least $r\cdot 2^h$ elements in interval $I$ that were inserted in the ensemble. we distribute the charge on $I$ equally among these elements; since $I$ is of length $\ell / 2^{H-h}$, each element gets a charge of $\frac{r\cdot \ell / 2^{H-h}}{r\cdot 2^h} = \ell / 2^H$. Now, each element can be in at most $H-1$ intervals corresponding to heights $h = 1, 2, \ldots, H-1$ that give it a charge, since the interval at height $H$ can only label root nodes and therefore, does not have any charge. Finally, we add the charge at leaf nodes: since only occupied leaf nodes can be charged, each element gets an additional charge of at most $\ell/2^H$. The final charge on each element is at most $H\cdot \ell / 2^H$, i.e., $nH\cdot (\ell / 2^H)$ overall.

\bigskip
    
    Now, we consider \two cost. Note that this is limited to the composite tree $A$. We differentiate between $x_i > x_{i+1}$ and $x_i \le x_{i+1}$, and define the {\em error} as $\max(x_i - x_{i+1}, 0)$. First, we bound error of \two. Later, we will show that this can be used to easily bound the cost of \two. 
    
    Let $x_i, x_{i+1}$ be in elementary trees corresponding to nodes $u_\ell, u_r$ in $A$. Let $v$ denote the least common ancestor of $u_\ell, u_r$ in $A$, and let $v_\ell, v_r$ be the left and right child node of $v$ in $A$. There are three subcases:
    \begin{itemize}
        \item[-] In the first subcase, $u_\ell, u_r$ are distinct from $v$. Then, $x_i$ is the rightmost non-empty leaf in the composite subtree under $v_\ell$, $x_{i+1}$ is the leftmost non-empty leaf in the composite subtree under $v_r$, and the elementary tree at node $v$ is empty. Note that by \Cref{cl:label}, $x_i$ is in the interval label of $v_\ell$ and $x_{i+1}$ is in the interval label of $v_r$. Since the interval label of $v_r$ is disjoint from and greater than the interval label of $v_\ell$, we have that $x_{i+1} > x_i$. Therefore, there is no \two error in this case.
        \item[-] In the second subcase, $u_\ell = v$. Then, $x_i$ is the rightmost non-empty leaf in the elementary subtree at node $v$ and $x_{i+1}$ is the leftmost non-empty leaf in the composite subtree under $v_r$. 
        \item[-] In the third subcase, $u_r = v$. This is symmetric to the second subcase. Here, $x_i$ is the rightmost non-empty leaf in the composite subtree under $v_\ell$ and $x_{i+1}$ is the leftmost non-empty leaf in the elementary subtree at node $v$. 
    \end{itemize}
    In the second and third subcases, consider the case when $x_i > x_{i+1}$ (i.e., there is non-zero error). The error $x_i - x_{i+1}$ is at most the length of the interval label of $v$ (call it $I_v$). We charge $x_i - x_{i+1}$ to node $v$. Any node in $A$ can be charged at most twice by this process, once each by the second and third subcases. So, if the length of $I_v$ is $\ell_v$, then the charge on $v$ is at most $2\ell_v$. Moreover, if $v$ is at height $h$ in $A$, then $\ell_v = \ell / 2^{H-h}$ and the charge on $v$ is at most $2 \ell / 2^{H-h}$. By \Cref{cor:composite-space}, at least $2^{h-1}$ elements in $I_v$ have been inserted in $A$. Distributing the charge on $v$ equally among these elements gives each element a charge of $\frac{2 \ell / 2^{H-h}}{2^{h-1}} = 4\ell / 2^H$. Finally, note that an element can be charged at most once by each height $h = 1, 2, \ldots, H$, which gives a cumulative charge of at most $4H\cdot (\ell / 2^H)$ to each element. In other words, the total error of \two is at most $4nH \cdot (\ell / 2^H)$.

    We now translate \two error to \two cost. Let $C_1, C_2$ respectively denote \one and \two costs in $A$, and $E$ denote the \two error in $A$. Summing over all elements inserted in $A$, we have
    \[
        C_1 + C_2 
        = \sum_i |x_{i+1} - x_i| 
            \le 2 \cdot \sum_{i: x_i > x_{i+1}} (x_i - x_{i+1}) + (\max_i x_i - \min_i x_i)
            \le 2 (C_1 + E) + \ell.
    \]
     Then, 
     \[
        C_2 \le C_1 + 2 E + \ell \le nH \cdot (\ell/2^H) + 8 nH \cdot (\ell/2^H) + \ell = 9 nH \cdot (\ell/2^H) + \ell.
    \]        
    
\bigskip        

    Finally, the bound on \three cost simply follows by observing that there are at most $\alpha$ pairs of elements that incur cost of this type, and any such pair incurs a cost of at most $\ell$ since all the elements belong to the interval label of the roots of the trees.
\end{proof}


\section{An $O(\frac{n \log n}{\alpha})$-space, $O(\alpha \log n)$-cost Data Structure} \label{sec:small-error-part1}
We give the data structures that achieve $\gamma\cdot c = O(\log^2 n)$ for $\gamma\in [O(1), O(\log^2 n)]$ in two parts. In this section, we consider the range $\gamma \in [O(1), O(\log n)]$, and in the next section, we consider $\gamma \in [O(\log n), O(\log^2 n)]$. So, for any $\alpha\in [O(1), O(\log n)]$, our goal in this section is to design an $O(\frac{n \log n}{\alpha})$-space, $O(\alpha \log n)$-cost data structure.

\subsection{Description of the Data Structure} 
Our data structure is a tree ensemble $A$ of order $\alpha$, i.e., it comprises two parts: a ``prefix'' that is a single composite tree $B$, and a ``suffix'' that comprises of $\alpha$ elementary trees $C_1, \cdots, C_\alpha$. The interval labels of the roots of all the trees $B, C_1, C_2, \ldots, C_\alpha$ is $[0, 1)$, and their heights are $H := \log (n / \alpha)$.

\begin{remark}
    \label{rem:med-size-trees-are-complete}
    We assume that all the elementary trees and the composite tree are complete binary trees. This is w.l.o.g.\footnote{without loss of generality} by rounding up $n$ to the next power of $2$, which increases the space by at most a factor of 2. We also assume w.l.o.g. that $\alpha$ is a power of 2 (otherwise round either up or down to the next power to $2$).
\end{remark}

\subsection{The Online Sorting Algorithm}

When presented a new element $x_t$ at time $t$, the online algorithm first tries to insert it in the prefix, i.e., in the composite tree $B$ using the algorithm in \Cref{sec:composite}. If it succeeds, the algorithm exits with success. Otherwise, it tries to inserts $x_t$ in $C_1, C_2, \ldots, C_\alpha$ in this order, using the algorithm in \Cref{sec:elementary}. The first time that the insertion succeeds, the algorithm exits with success. 

We show that the algorithm succeeds to insert each element, i.e., it never happens that the algorithm fails for all of $B, C_1, C_2, \ldots, C_\alpha$.
%
For the purpose of the next proof, we consider the root of an elementary tree unmarked until some element has been inserted in it. 
\begin{claim}\label{cl:med-correct}
    Every element gets inserted by the online algorithm.
\end{claim}
\begin{proof}
    We claim that the algorithm never marks the root of the elementary tree corresponding to $C_\alpha$ in the suffix. Suppose it does. Then, including the root node of the composite tree, there are $\alpha+1$ nodes at height $H$ that are labeled by the interval $[0, 1)$. By \Cref{lem:main-space}, this implies that $\alpha \cdot 2^H + 1$ elements in $[0, 1]$ have been inserted into the array. But, 
    \[
        \alpha\cdot 2^H + 1 
        = \alpha \cdot 2^{\log (n/\alpha)} + 1
        = n + 1
        > n,
    \]
    which is a contradiction.
\end{proof}

\subsection{Space Bound}

Since $H = \log (n/\alpha)$, according to \cref{lem:space-composite} the size of the prefix which includes the single composite tree is $2^{\log (n/\alpha)} \cdot \log (n/\alpha) = (n / \alpha) \cdot \log (n/\alpha)$. Each elementary tree $C_k$ in the suffix is of size $2^{\log(n/\alpha)} = n/\alpha$ and there are $\alpha$ such trees. Hence, the total size of the suffix is $n$. Putting everything together, the total size of the array $A$ is given by
\[
   (n/\alpha) \cdot \log (n/\alpha) + n = O\left(\frac{n\log n}{\alpha}\right), \text{ since } \alpha \le \log n.
\]

\subsection{Cost Bound}
To bound the cost, we use \cref{lem:main-cost}. The length of the interval labeling the roots is $\ell = 1$ and the height $H = \log (n/\alpha)$, i.e., $2^H = n/\alpha$. By \Cref{lem:main-cost}, the costs of \one, \two, and \three are respectively 
$n H\cdot (\ell / 2^H) = \alpha \log (n/\alpha)$, $9 nH\cdot (\ell/2^H) + \ell = 9\alpha \log (n/\alpha) + 1$, and $\alpha \ell = \alpha$. These add to $O(\alpha \log{n})$.

\section{An $O(\beta n \log n)$-space, $O\left(\frac{\log n}{\beta}\right)$-cost Data Structure}\label{sec:small-error-part2}
In this section, we give the data structure that achieves $\gamma\cdot c = O(\log^2 n)$ for the range $\gamma \in [O(\log n), O(\log^2 n)]$. Along with the data structure in \Cref{sec:small-error-part1}, this establishes $\gamma\cdot c = O(\log^2 n)$ for the entire range of interest $[O(1), O(\log^2 n)]$. So, for any $\beta\in [O(1), O(\log n)]$, our goal in this section is to design an $O(\beta n \log n)$-space, $O\left(\frac{\log n}{\beta}\right)$-cost data structure.

\subsection{Description of the Data Structure} 
Our data structure is a sequence of $\beta$ composite trees $A_1, \ldots, A_\beta$ each of height $H :=  \log n$. The interval label of the $k$th tree is $\left[\frac{k-1}{\beta}, \frac{k}{\beta}\right)$. The next remark is akin to \Cref{rem:med-size-trees-are-complete}:

\begin{remark}
    \label{rem:large-size-trees-are-complete}
    We assume that all the elementary trees and the composite tree are complete binary trees. This is w.l.o.g. by rounding up $n$ to the next power of $2$, which increases the space by at most a factor of 2. We also assume w.l.o.g. that $\beta$ is a power of 2 (otherwise round either up or down to the next power to $2$).
\end{remark}

\subsection{The Online Sorting Algorithm}

Suppose the online algorithm is presented a new element $x_t$ at time $t$. First, we determine the composite array $A_k$ that the element can be inserted in. Note that the intervals associated with these composite arrays are disjoint, and hence, there is a unique choice of $k$ given a fixed element $x_t$. 
The insertion of $x_t$ into the appropriate composite tree is as described in \ref{sec:composite}. 

We now show that the algorithm succeeds in inserting every element. For the purpose of the next proof, as in \Cref{cl:med-correct}, we consider the root of an elementary tree unmarked until some element has been inserted in it. 
\begin{claim}\label{cl:large-correct}
    Every element gets inserted by the online algorithm.
\end{claim}
\begin{proof}
    Suppose the algorithm were unable to insert an element $x_t$. Imagine that we insert the element in a second copy of the elementary tree at the root of the composite tree. Then, the roots of the two elementary trees of height $H$, one at the root of the composite tree and the other in the elementary tree where the last element is inserted, are marked nodes at height $H$ that are labeled by the same interval. By \Cref{lem:main-space}, this implies that $2^H + 1$ elements have been inserted from this interval. But, 
    \[
        2^H + 1 
        = 2^{\log n} + 1
        = n + 1
        > n,
    \]
    which is a contradiction.
\end{proof}

\subsection{Space Bound}

Since $H = \log n$, each composite tree has space $n \log n$ by \Cref{lem:space-composite}. Therefore, the overall space is $\beta n\log n$.

\subsection{Cost Bound}
First, consider the cost between the rightmost element of a composite tree and the leftmost element of the next composite tree. Since the interval labels of the roots of these composite trees are of length $1/\beta$ and these intervals are adjacent to one another, the maximum difference between these two elements can be $2/\beta$. Adding over the $\beta-1$ adjacent pairs, this gives a total cost of at most $2$ between elements occupying adjacent non-empty cells in different composite trees.

We now bound the cost within a composite tree using \Cref{lem:main-cost}. Note that a single composite tree is a tree ensemble of order $0$ and has only \one and \two cost. The length of the interval labeling the root is $\ell = 1/\beta$ and the height $H = \log n$, i.e. $2^H = n$. Let $n_i$ be the number of elements inserted in the $i$th composite tree $A_i$. By \Cref{lem:main-cost}, the costs of \one and \two for $A_i$ are $n_i H\cdot (\ell/2^H) = \frac{n_i \log n}{n \beta}$ and $9n_i H\cdot (\ell/2^H) + \ell = \frac{9 n_i \log n}{n \beta} + \frac{1}{\beta}$ respectively. These add to $O\left(\frac{n_i \log n}{n \beta} + \frac{1}{\beta}\right)$. Since $\sum_{i=1}^\beta n_i = n$, the total cost across the $\beta$ composite trees is at most $O\left(\frac{\log n}{\beta} + 1\right)$, which is $O\left(\frac{\log n}{\beta}\right)$ since $\beta = O(\log n)$.  

\section{Unknown Optimal Cost}\label{sec:doubling}
In this section, we consider the scenario where the range of the elements (and therefore, cost of the optimal solution $\opt$) is not known in advance. Until now, our analysis has focused on the case when the range of elements is known, which is $[0, 1]$ w.l.o.g. by scaling. In this context, the cost corresponds to the competitive ratio since $\opt = 1$. 
In this section, we show that we can recover all the results for known range even when the range of elements is unknown, using the same space and with an $O(\log n)$ overhead in the competitive ratio. 

We prove the following two theorems:
\begin{theorem}
    \label{thm:uknown-small-space}
    There exists an $O\left(\frac{\log^3 n}{\eps}\right)$-competitive algorithm using $(1+\eps)n$ space, for any $\eps \in [3\log n /n,1]$.
\end{theorem}
\begin{theorem}
    \label{thm:uknown-larger-space}
    There exists an $O\left(\frac{\log^3 n}{\gamma}\right)$-competitive algorithm using $O(\gamma n)$ space, for any $\gamma \in [1,\log^2 n]$.
\end{theorem}

To adapt the prior data structure for known range to this setting, we apply a doubling scheme to two parameters: the optimal cost $\opt$ and the size of the data structure for a given guess on $\opt$. We have an outer loop that (at least) doubles the estimate on $\opt$ in each iteration, and an inner loop that doubles the number of inserted elements. 
This nested doubling process is initialized after sequentially inserting the first two elements (in the first two array cells), whose difference provides an initial estimate for $\opt$. 


Each iteration of the outer loop is called an epoch. In the $i$-th epoch, the guess on $\opt$ is denoted $\opt_i$, and corresponds to a range $[a_i, b_i]$. When a new element arrives outside this interval, the $i$-th epoch ends and the $i+1$-st epoch starts. The range in the new epoch is at least double that in the previous epoch, and includes the new element. E.g., if the new element is $c > b$, the interval is updated to $[a_{i+1}, b_{i+1}]$ where $a_{i+1} := a_i$ and $b_{i+1} := \max\{c, 2b_i-a_i\}$. (The case $c < a$ is symmetric.) Correspondingly, $\opt_{i+1}$ is updated to $b_{i+1} - a_{i+1}$, which is at least $2\cdot \opt_i$. Each epoch is assigned a newly allocated region of space, which begins immediately after the last space used in the previous epoch.


Within an epoch, each iteration of the inner loop is called a phase.
Let $n_j$ denote the bound on the number of inserted elements in the $j$-th phase. The phase ends when the number of elements exceeds $n_j$, and the $j+1$-st phase starts with $n_{j+1} := 2 n_j$. Overall, we define $n_j := 2^j$.
There are some subtle differences in the space allocation in a phase between two cases: (a) for $(1 + \varepsilon)n$ space corresponding to \Cref{sec:small-space}, and (b) for at least linear in $n$ space corresponding to \Cref{sec:small-error-part1,sec:small-error-part2}.
We discuss these separately below. 

\subsection{$(1+\eps)n$ space, $\eps \in [3\log n /n,1] $}

In the $j$-th phase, we allocate a new array that begins immediately after the array used in the $j-1$-st phase. The amount of space allocated in the $j$-th phase is $(1+\eps/3)n_j$.
During the $j$-th phase, we employ the algorithm from \Cref{sec:small-space} for inserting $n_j$ elements with the variable $\eps / 3$ (instead of $\eps$) as a black box.
Since elementary trees are not defined for small values of $n$ (see \Cref{rem:epsilon-cannot-be-too-small}), we start any epoch by inserting the first $\frac{\log^2{\eps}}{\eps}$ elements consecutively before starting the doubling scheme.

In the proofs below, we use \(k_i, n_i\) to respectively denote the number of phases and elements in the $i$-th epoch, and \(n_{i,j}\) to denote the number of elements in its \(j\)-th phase.

\begin{lemma}[Space Bound]\label{lem:unknown-small-space-space}
    The algorithm inserts $n$ elements in $(1+\eps)n$ space.
\end{lemma}
\begin{proof}
By \Cref{lem:smallspace-unused-space}, in phases $\le k_i-1$, the unused space is at most $\eps/3$ times the used space, and in phase $k_i$, it is at most $\eps/3$ times $n_{i, k_i}$. But, $n_{i, k_i} \le 2 n_{i, k_{i-1}}$, i.e., the unused space in phase $k_i$ is at most $2 \eps/3$ times the used space in this epoch. Thus, the total unused space is at most an $\eps$-fraction of the used space in this epoch. The lemma follows. 
\end{proof}

\begin{lemma}[Cost Bound]
    The cost incurred by the algorithm is at most $\opt\cdot O(\log ^3 n / \eps)$.
\end{lemma}
\begin{proof}
Fix epoch $i$. The cost between consecutive phases is at most \(\opt_i\), based on the current bound on the optimal solution. This results in a total additive cost of at most \(k_i \cdot \opt_i\) between the $k_i$ phases. In addition, since we insert the first $\frac{\log^2 (1/\eps)}{\eps}$ elements sequentially, we incur an additional cost of at most $\frac{\log ^2 (1/\eps)}{\eps} \cdot \opt_i$. Finally, there is a cost of $O\left(\frac{\log^2 n_{i, j}}{\eps}\right)\cdot \opt_i = O(j^2/\eps)\cdot \opt_i$ within each phase by \Cref{lem:cost-smallspace}.
Therefore, the total cost in the \(i\)-th epoch is at most
\[
\opt_i \cdot O\left(\frac{\log^2 (1/\eps)}{\eps} + k_i + \sum_{j=1}^{k_i} \frac{j^2}{\eps} \right)
= \opt_i \cdot O\left(\frac{\log^2 (1/\eps)}{\eps} + \frac{k_i^3}{\eps} \right). 
\]
In addition, the cost between epochs $i-1$ and $i$ is at most $\opt_i$.
Summing over all $\ell$ epochs, we get that the total cost is at most
\[
\sum_{i=1}^\ell \opt_i \cdot O\left(\frac{\log^2 (1/\eps)}{\eps} + \frac{k_i^3}{\eps} \right).
\]

Note that $\opt_i$ at least doubles in every epoch, i.e., $\sum_{i=1}^\ell \opt_i \le 2 \opt_\ell \le 2\opt$. Moreover, the total number of phases in any epoch is at most $\log n$, i.e., $k_i \le \log n$ for every $i$. Applying these facts to the above cost bound, we get that the total cost is at most
\[
    \opt \cdot O\left(\frac{\log^2 (1/\eps)}{\eps} + \frac{\log^3 n}{\eps} \right)
    = \opt \cdot O(\log^3 n / \eps).\qedhere
\]
\end{proof}

\eat{
Now consider the total error over \(\ell\) epochs. First, we bound the error between epochs. The cost from epoch \(i-1\) to epoch \(i\) is at most \(\opt_i\). Since the bound on \(\opt\) at least doubles in each epoch, we have
\[
\sum_{i=1}^{\ell} \opt_i \leq 2 \, \opt_\ell \leq 2 \, \opt.
\]

Now, consider \(\ell\) epochs of the doubling scheme on the value of the optimal solution. Since the bound on the value of \(\opt\) at least doubles in each epoch, and the algorithm is competitive relative to the given value of \(\opt\), the cost in the \(i\)-th epoch is scaled by a factor of \(1 / 2^{\ell - i}\). 
Then the error in the \(i\)-th epoch is:
\[
\frac{\frac{\log ^2 (\frac{1}{\eps})}{\eps} + k_i + \sum_{j=1}^{k_i} f(2^j)}{2^{\ell - i}}.
\]

Since the bound on the value of \(\opt\) in the final (\(\ell\)-th) epoch is at most \(2 \cdot \opt\), the total cost over all epochs is bounded by:
\[
2\opt \cdot \left( \sum_{i=1}^{\ell} \frac{\frac{\log ^2 (\frac{1}{\eps})}{\eps}+ k_i + \sum_{j=1}^{k_i} f(2^j)}{2^{\ell - i}} \right)
\]

Recall that \(f(n) = O(\frac{\log^2 n}{\eps}\)). We obtain the final bound on the total cost:
\begin{align*}
&  O\left( \opt \cdot \sum_{i=1}^{\ell} \frac{\frac{\log ^2 (\frac{1}{\eps})}{\eps} + k_i + \sum_{j=1}^{k_i} \frac{\log^2 {2^j}}{\eps}}{2^{\ell - i}} \right) = O\left( \opt \cdot \sum_{i=1}^{\ell} \frac{\frac{\log ^2 (\frac{1}{\eps})}{\eps} + k_i+ \sum_{j=1}^{k_i} \frac{\j^2}{\eps}}{2^{\ell - i}} \right)
 \\
 & = O\left( \opt \cdot  \sum_{i=1}^{\ell} \frac{\frac{\log ^2 (\frac{1}{\eps})}{\eps}+ k_i +\frac{1}{\eps} \sum_{j=1}^{k_i} j^2}{2^{\ell - i}} \right) 
= O\left( \opt \cdot  \sum_{i=1}^{\ell} \frac{\frac{\log ^2 (\frac{1}{\eps})}{\eps}+ k_i + \frac{k_i^{3}}{\eps}}{2^{\ell - i}} \right) \\
&= O\left( \frac{\opt}{\eps} \cdot  \sum_{i=1}^{\ell} \frac{\log ^2 (\frac{1}{\eps})+k_i^{3}}{2^{\ell - i}} \right)
\end{align*}
\begin{remark}
    \label{rem:unknown-bound-on-elements-iterations-small-space}
    The number of phases is at most \(\log n\), as this suffices to insert up to \(n\) elements.
\end{remark}
We rewrite the important term: 
\(
\sum_{i=1}^{\ell} \frac{k_i^{3}}{2^{\ell - i}}
\), using \cref{rem:unknown-bound-on-elements-iterations-small-space} the term is bounded by:
\[
\sum_{i=1}^{\ell} \frac{\log ^2 (\frac{1}{\eps})+\log ^3 n}{2^{\ell - i}}  
\le 2\log ^3 n\sum_{i=1}^{\ell} \frac{1}{2^{\ell - i}} \le 4\log ^3 n
\]
We thus obtain the desired bound:
\[
O\left(\opt \cdot \frac{\log ^3 n}{\eps}\right)
\]
}

\subsection{O($\gamma n$) space, $\gamma \in [1, \log^2{n}]$}

In the $j$-th phase, we allocate a new array that begins immediately after the array used in the \((j-1)\)-st phase. The allocated space in the $j$-th phase is $\frac{n_j \log n_j}{\alpha}$ or $\beta n_j \log n_j$, depending on whether the total space is $\frac{n \log n}{\alpha}$ or $\beta n \log n$. The parameters $\alpha$ and $\beta$ are as defined in \Cref{sec:small-error-part1,sec:small-error-part2} respectively.
In phase \(j\), we employ the algorithm designed for handling up to \(2^j\) elements as a black box, using the appropriate algorithm from \Cref{sec:small-error-part1,sec:small-error-part2}.

In the proofs below, we use \(k_i, n_i\) to respectively denote the number of phases and elements in the $i$-th epoch, and \(n_{i,j}\) to denote the number of elements in its \(j\)-th phase.

\begin{lemma}[Space Bound]\label{lem:unknown-small-error-space}
    The algorithm inserts $n$ elements in $O(\gamma n)$ space.
\end{lemma}
\begin{proof}
Since $n_{i, k_i} = 2 n_{i, k_i-1} \le 2 n_i$ and $n_{i, j}$ doubles in every phase, the total space used in the $i$-th epoch is at most 
\[
    \sum_{j=1}^{k_i} \gamma\, n_{i, j} 
    \le 2 \gamma\, n_{i, k_i} 
    \le 4 \gamma\, n_i.
\]
The lemma follows since $\sum_i n_i = n$.
\end{proof}

\eat{
In phases $\le k_i-1$, the unused space is at most $\eps/3$ times the used space, and in phase $k_i$, it is at most $\eps/3$ times $n_{i, k_i}$. But, $n_{i, k_i} \le 2 n_{i, k_{i-1}}$, i.e., the unused space in phase $k_i$ is at most $2 \eps/3$ times the used space in this epoch. Thus, the total unused space is at most an $\eps$-fraction of the used space in this epoch. The lemma follows. 

For the first \(k_i - 1\) phases, the total space used is at most \(\sum_{j=1}^{k-1} \alpha n_{i,j}\), as guaranteed by the correctness of the algorithms in \cref{cl:med-correct} and \cref{cl:large-correct}. In the \(k_i\)-th phase, some portion of the allocated space may remain unused. This unused space is at most \(2n_{i,j-1} \cdot \alpha\).
Therefore, the total space used for inserting all elements is bounded by :
\[2\alpha n_{i,j-1} + \sum_{j=1}^{k_i} \alpha n_{i,j} \leq 3\sum_{j=1}^{k_i} \alpha n_{i,j} = 3\alpha  n_i
\]
As this argument applies across all epochs, we conclude that all elements can be inserted into the data structure within $3\alpha n $ 
space.
}

\begin{lemma}[Cost Bound]
    The cost incurred by the algorithm is at most $\opt\cdot O(\log ^3 n / \gamma)$.
\end{lemma}
\begin{proof}
Fix epoch $i$. The cost between consecutive phases is at most \(\opt_i\), based on the current bound on the optimal solution. This results in a total additive cost of at most \(k_i \cdot \opt_i\) between the $k_i$ phases. In addition, there is a cost of $O\left(\frac{\log^2 n_{i, j}}{\gamma}\right)\cdot \opt_i = O(j^2/\gamma)\cdot \opt_i$ within each phase.
Therefore, the total cost in the \(i\)-th epoch is at most
\[
\opt_i \cdot O\left(k_i + \sum_{j=1}^{k_i} j^2 / \gamma \right)
= \opt_i \cdot O\left(k_i + k_i^3 / \gamma \right). 
\]
In addition, the cost between epochs $i-1$ and $i$ is at most $\opt_i$. Inserting the first two elements in the $i$-th  So, epoch sequentially results in a further cost of $2 \opt_i$. So, these add $3 \opt_i$ to the cost of epoch $i$.
Summing over all $\ell$ epochs, we get that the total cost is at most
\[
\sum_{i=1}^\ell \opt_i \cdot O\left(3 + k_i + k_i^3 / \gamma \right)
= \sum_{i=1}^\ell \opt_i \cdot O\left(k_i + k_i^3 / \gamma \right).
\]

Note that $\opt_i$ at least doubles in every epoch, i.e., $\sum_{i=1}^\ell \opt_i \le 2 \opt_\ell \le 2\opt$. Moreover, the total number of phases in any epoch is at most $\log n$, i.e., $k_i \le \log n$ for every $i$. Applying these facts to the above cost bound, we get that the total cost is at most
$\opt \cdot O(\log n + \log^3 n/\gamma)$. Since $\gamma \le O(\log^2 n)$, the lemma follows.
\end{proof}

\eat{    

    Between each of these \(k_i\) phases, the incurred cost is at most \(\opt_i\), based on the current bound on the optimal solution. This results in a total additive cost of at most \(k_i \cdot \opt_i\). In addition, in each epoch we start by inserting the first 2 elements sequentially for an extra total cost of at most $2 \cdot \opt_i$.
In each of the \(k_i\) phases, the cost corresponds to the algorithm used under the corresponding space constraint. Let \(f(n)\) denote the cost incurred when inserting \(n\) elements within the given space regime.

Therefore, the total cost within the \(i\)-th epoch is bounded by:
\[
\opt_i \cdot \left(2 + k_i + \sum_{j=1}^{k_i} f(2^j) \right)
\]
where \(f(2^j)\) represents the cost in the \(j\)-th phase handling up to \(2^j\) inserted elements.
Now consider the total cost over \(\ell\) epochs in the doubling process on the value of the optimal solution. First, we bound the cost between epochs. The cost from epoch \(i-1\) to epoch \(i\) is at most \(\opt_i\). Since the bound on \(\opt\) at least doubles in each epoch, we have:
\[
\sum_{i=1}^{\ell} \opt_i \leq \sum_{i=1}^{\ell} \frac{\opt_\ell}{2^{\ell - i}} \leq 2 \opt_\ell \leq 2 \opt.
\]
This inter-epoch contribution adds only a constant-factor overhead and can therefore be subsumed in the overall bound.

Next, we consider the within-epoch cost, noting that the cost in the \(i\)-th epoch is scaled by a factor of \(1 / 2^{\ell - i}\).
Therefore, the cost accumulated in the \(i\)-th epoch is:
\[
\frac{2+k_i + \sum_{j=1}^{k_i} f(2^j)}{2^{\ell - i}},
\]
The total cost across all epochs is bounded by:
\[
2\opt \cdot \left( \sum_{i=1}^{\ell} \frac{2+k_i + \sum_{j=1}^{k_i} f(2^j)}{2^{\ell - i}} \right),
\]

 Recall that \(f(n) = O(\frac{\log^2 n}{\alpha})\). We obtain the final bound on the total cost:
\begin{align*}
& O\left( \opt \cdot \sum_{i=1}^{\ell} \frac{k_i + \sum_{j=1}^{k_i} \frac{\log^2 {2^j}}{\alpha}}{2^{\ell - i}} \right) 
= O\left( \opt \cdot \sum_{i=1}^{\ell} \frac{k_i + \sum_{j=1}^{k_i} \frac{\j^2}{\alpha}}{2^{\ell - i}} \right)
= O\left( \opt \cdot  \sum_{i=1}^{\ell} \frac{k_i +\frac{1}{\alpha} \sum_{j=1}^{k_i} j^2}{2^{\ell - i}} \right)\\
 & = O\left(\opt \cdot  \sum_{i=1}^{\ell} \frac{k_i + \frac{k_i^{3}}{\alpha}}{2^{\ell - i}} \right) = O\left(\frac{\opt}{\alpha} \cdot  \sum_{i=1}^{\ell} \frac{k_i^{3}}{2^{\ell - i}} \right)
\end{align*}
\begin{remark}
    \label{rem:unknown-bound-on-elements-iterations}
    The number of phases is at most \(\log n\), as this suffices to insert up to \(n\) elements.
\end{remark}
We rewrite the important term: 
\(
\sum_{i=1}^{\ell} \frac{k_i^{3}}{2^{\ell - i}}
\), using \cref{rem:unknown-bound-on-elements-iterations} the term is bounded by:
\[
\sum_{i=1}^{\ell} \frac{\log ^3 n}{2^{\ell - i}} = \log ^3 n\sum_{i=1}^{\ell} \frac{1}{2^{\ell - i}} \le 2\log ^3 n
\]
We thus obtain the desired bound:
\[
O\left(\opt \cdot \frac{\log ^3 n}{\alpha}\right)
\]

}

\section{Closing Remarks}\label{sec:closing}
In this paper, we gave nearly tight bounds for the online sorting problem. Specifically, we gave deterministic algorithms with (a) $(1+\eps)n$ space and $O(\log^2 n / \eps)$ cost for any $\eps = \Omega(\log n / n)$, and (b) $\gamma n$ space and $O(\log^2 n / \gamma)$ cost for any $\gamma \in [O(1), O(\log^2 n)]$. These results are exponential improvements over the best previous results for the problem. Coupled with an existing lower bound of $\gamma\cdot c = \Omega(\log n)$, our results nearly resolve the online sorting problem. We also show that our results extend to the case where the range of the elements is not known in advance, with an additional factor of $O(\log n)$ in the competitive ratio of the algorithms.
\
{\small
\bibliography{bib}
}

\end{document}